\documentclass[12pt,a4paper]{amsart}

\usepackage{xypic}

\newtheorem{thm}{Theorem}[section]

\newtheorem{defin}[thm]{Definition}

\newtheorem{remark}[thm]{Remark}

\newcommand\Img{\operatorname{\mathcal Im}}

\newcommand\iso{\kern.35em{\raise3pt\hbox{$\sim$}\kern-1.1em\to}\kern.3em}

\newcommand\F{\mathbb F}

\newcommand\Pp{\mathbb P}
\newcommand\m{\mathfrak m}
\newcommand\Cc{\mathcal C}

\newcommand\Oc{\mathcal O}

\newcommand\Res{\operatorname {Res}}

\newcommand\Proj{\operatorname {Proj}}
\newcommand\df{\operatorname {\it d_{free}}}
\newcommand\ord{\operatorname {order}}

\begin{document}
\title{One dimensional Convolutional Goppa Codes over the projective line }
\author{J.A. Dom\'{\i}nguez P\'erez, J.M. Mu\~noz Porras and
G. Serrano Sotelo}
\thanks {This work was partially supported by the research contracts MTM2009-11393 of the Spanish Ministry for Science  and Innovati\'on.}
\thanks {J.A.~Dom\'{\i}nguez~P\'erez, J.M.~Mu\~noz~Porras and
G.~Serrano~Sotelo are in the Department of Mathematics, University of Salamanca, Plaza de la Merced 1-4, 37008 Salamanca, Spain (email: jadoming@usal.es, jmp@usal.es,  and laina@usal.es)}


\maketitle

\begin{abstract}

We give a general method to construct MDS one-dimensional convolutional codes. Our method generalizes previous constructions \cite{HGBL06}. Moreover we give a classification of one-dimensional Convolutional Goppa Codes and propose a characterization of MDS codes of this type.

\end{abstract}

\section*{Introduction}
One of the main problems in coding theory is the construction of codes with a large distance, such as so-called MDS codes.

The aim of this paper is to give a very general method to construct one-dimensional 
MDS convolutional codes using the techniques developed in our previous papers \cite{DMPS06,DMPS04,DMS:08}.

In Section \ref{sec:conv} we give a general introduction to convolutional codes, reformulated in terms that enables a good understanding of the choices of the generator matrices and the submodules generated by them. The treatment is fairly self-contained, with only a few references for proofs of certain statements. Moreover,  a characterization of one-dimensional MDS convolutional codes in terms of their associated block linear codes is given (Theorem \ref{linearG}).

In Section \ref{sec:goppa} we describe the notion of Convolutional Goppa Code, introduced in \cite{DMPS06} and  \cite{DMPS04} and we recall the construction of convolutional Goppa codes over the projective line. We use this construction in Section \ref{sec:convgoppa} to give families of examples of one-dimensional convolutional Goppa codes; moreover, we prove, using Theorem \ref{linearG}, that they are MDS. This is the main result in this paper. The examples constructed in \cite{HGBL06} are particular cases of ours.

Finally, we give in Section \ref{sec:class} a classification of one-dimensional convolutional Goppa codes defined over the projective line, which could give rise to a characterization of MDS convolutional Goppa codes of dimension one.

\section{Convolutional codes}\label{sec:conv}

Given a finite field $\F_q$, representing the symbols in which an
information word $u\in \F_q^k$ is written, each $k\times n$ matrix of rank $k$ with entries in $\F_q$ defines an injective linear map
$$
\begin{aligned}\F_q^k &\xrightarrow{\mathcal G} \F_q^n
\\
u&\mapsto x=u\mathcal G\,,\end{aligned}
$$
whose image subspace is the linear code $\Cc=\Img \mathcal G\subseteq
\F_q^n$ of length  $n$, dimension $k$, and rate $k/n$. $\mathcal G$ is called a generator matrix of the code, and $\mathcal G'$ is another generator matrix of the code if there exists an element $B\in GL(k,\F_q)$ such that  $\mathcal G'=B\cdot\mathcal G$.

In practical applications, the codification process is not limited
to a single word, but to a sequence of information words depending on
time, $u_t\in \F_q^k$, $t\geq0$, which after the codification
are transformed into the sequence of codified words
$x_t=u_t \mathcal G$
and $x_t$ at the instant $t$ depends only on the
information word $u_t$ at the same instant $t$.

The basic idea of \emph{convolutional codification} is to allow $x_t$ to
depend not only on $u_t$ but also on $u_{t-1},\dots, u_{t-m}$ for
some positive integer $m$, which is the \emph{memory} of the code.
If one denotes a sequence
of words as a polynomial vector
$
u(z)=\displaystyle\sum_{t=0}u_t z^t \in \F_q[z]^k
$
and the product by $z^i$  as a \emph{delay
operator},
$
z^i u(z)=\displaystyle\sum_{t=0} u_t z^{t+i} = \displaystyle\sum_{t=i} u_{t-i} z^{t}\,,
$
each $k\times n$ matrix $\mathcal G$ of rank $k$ with entries in $\F_q[z]$ defines an injective morphism 
 of $\F_q[z]$-modules
$$
\begin{aligned}\F_q[z]^k & \xrightarrow{\mathcal G}\F_q[z]^n\\
u(z)&\mapsto x(z)=u(z)\mathcal G\end{aligned}
$$
and one says that the image submodule is a \emph{convolutional code},
$\Cc=\Img \mathcal G\subseteq
\F_q[z]^n$, of length $n$ and dimension $k$,  and that $\mathcal G$ is a generator matrix of $\Cc$.
One can then define \emph {a rate $k/n$ convolutional code $\Cc$ as a submodule of rank $k$ of $\F_q[z]^n$}.

If we allow the possibility of performing \emph{feedback}, then we can reverse the delay and define convolutional
codification  over  the field of fractions, $\F_q(z)$,
of $\F_q[z]$.

\begin{defin}
A rate $k/n$ \emph{convolutional code} $\Cc$ over $\F_q$ is a $\F_q(z)$-linear
subspace of dimension $k$  of  $\F_q(z)^n$. The integers $(n,k)$ are called, respectively, the \emph{length} and
\emph{dimension} of the convolutional code.
\end{defin}

Each  \emph{generator matrix} $\mathcal G$ of $\Cc$ with entries in $\F_q(z)$ defines an 
injective linear \emph{encoding map}:
$$
\begin{aligned} \F_q(z)^k&\xrightarrow{\mathcal G} \F_q(z)^n\\
u(z)&\mapsto x(z)=u(z)\mathcal G\,,\text{ such that $\Img \mathcal G=\Cc$.}
\end{aligned}
$$
Given two generator matrices $\mathcal G$ and $\mathcal G'$ of the convolutional
code $\Cc$, there exists an element 
$B\in GL(k,\F_q (z))$ such that $\mathcal G'=B\cdot \mathcal G$.

If $\mathcal G$ and $\mathcal G'$ are \emph{polynomial generator matrices} of $\Cc$, that is, with entries in $\F_q[z]$, then their image submodules, as morphisms of $\F_q[z]$-modules  $ \F_q[z]^k\xrightarrow{\mathcal G,\mathcal G'} F_q[z]^n$, satisfy
$$
\Img\mathcal G\otimes_{\F_q[z]} \F_q(z)=\Cc= \Img\mathcal G'\otimes_{\F_q[z]} \F_q(z);
$$
although, they may be different: $\Img\mathcal G\neq \Img\mathcal G'$.

Thus, we are interested in polynomial generator matrices that define the same submodule. 
The family of the polynomial generator matrices called \emph{basic} (\cite{For70}, \cite{McE98}), satisfies this property.

\bigskip

\subsection{Basic generator matrices. Degree of a convolutional code}
\
\smallskip

Let $\Cc\subseteq \F_q(z)^n$ be a $(n,k)$ convolutional code.

\begin{defin}
A polynomial generator matrix $\mathcal G$ of $\Cc$, $ \F_q[z]^k\xrightarrow{\mathcal G} F_q[z]^n$,   is basic if any of the following
equivalent conditions are satisfied:
\begin{enumerate}
\item The quotient module $\F_q[z]^n/\Img\mathcal G$ is free.
\item The invariant factors of $\mathcal G$ are all equal to one.
\item The greatest common divisor of the order $k$ minors of $\mathcal G$ is equal to one.
\item $\mathcal G$ has a right inverse in $\F_q[z]$.
\end{enumerate}
\end{defin}

The existence of basic matrices for all convolutional codes was
proved in a constructive way by Forney \cite{For70}, using the Smith
algorithm for the computation of invariant factors (\cite{DMS:08},Theorem 11.16).

\begin{thm}\emph{(\cite{DMS:08})} Let $ \F_q[z]^k\xrightarrow{\mathcal G,\mathcal G'} F_q[z]^n$ be two polynomial generator matrices of $\Cc$. One has:
\begin{enumerate}
\item If $\mathcal G$ is basic and $\Img \mathcal G\subseteq \Img \mathcal G'$, then $\Img\mathcal G =\Img \mathcal G'$.
\item If $\mathcal G$ and $\mathcal G'$ are basic, then $\Img\mathcal G =\Img \mathcal G'$; that is, basic generator matrices define the same submodule.
\end{enumerate}
\end{thm}

\begin{thm} \emph{(\cite{DMS:08},\cite{McE98})}The family of the basic generator matrices is invariant under the action of the unimodular group $GL(k,\F_q[z])$. Thus, the number $$\delta_\mathcal G=\text{maximum degree of the order $k$ minors of  $\mathcal G$}$$
is the same for all basic encoders. 
\end{thm}

One can now consider an invariant associated with the code, namely, the degree of the code, which is defined as follows:

\begin{defin}\label{def:deg2} The degree $\delta$ of a convolutional
code $\Cc$ is
$$\delta=\delta_\mathcal G\,,\text{ where $\mathcal G$ is any basic encoder of $\Cc$.}$$
\end{defin}

\medskip

\subsection{Minimal basic generator matrices. Canonical matrices}
\
\medskip

In the implementation of convolutional codes as physical devices it
is convenient to find \emph{minimal encoders}, in the sense that the
corresponding circuit will have the minimum possible quantity of memory boxes. The
formalization of the concept of minimality can be expressed in terms
of the degree $\delta$ of the code.

If $\mathcal G$ is a polynomial generator matrix of $\Cc$, one denotes by $\deg \mathcal G$ the sum of its row degrees.

\begin{thm}  \cite{For70} For each $(n,k)$ convolutional code of degree $\delta$ there exists at least one basic generator matrix $\mathcal G$ such that
$$\delta=\deg \mathcal G\,.$$
Moreover, 
$$\deg \mathcal G\leq \deg \mathcal G'$$
for all polynomial encoders $\mathcal G'$ of the convolutional code.
\end{thm}

These basic generator matrices were called \emph{minimal basic matrices} by Forney \cite{For70} or \emph{canonical matrices} by McEliece  \cite{McE98}.

\subsection{Dual code. Control matrix}
\
\medskip

\label{def:dual} Given an $(n,k)$-convolutional code $\Cc\subseteq \F_q(z)^n$, the
dual code is the $\F_q(z)$-subspace defined by
$$
\Cc^\perp=\{ y(z)\in \F_q(z)^n \ /\  \langle x(z),y(z)\rangle=0
\text{\ for every $x(z)\in\Cc$}\}\,,
$$
with respect to the pairing
$\langle x(z) , y(z) \rangle = \sum_{i=1}^n x_i(z) y_i(z)$,
where $x(z)=(x_1(z),\dots,\break x_n(z))$ and $y(z)=(y_1(z),\dots,y_n(z))$ are in 
$\F_q(z)^n$.

\begin{thm}(\cite{DMS:08},Theorem 11.28 ) 
$\Cc^\perp$ is an $(k,n-k)$ \emph{convolutional code with the same degree as} $\Cc$.
\end{thm}

One defines a \emph{control matrix} for $\Cc$ as a $n-k\times n$ generator
matrix $H$ of its dual code $\Cc^\perp$.

\bigskip

\subsection{Weights and Free Distance}
\
\medskip

The \emph{(Hamming)
weight} of a vector $x=(x_1,\dots,x_n)\in\F_q^n$ is given by $w(x)=\#\{i\ |\ x_i\ne 0\}$ and the
\emph{(Hamming) distance} between $x,y\in\F_q^n$ is defined as
the weight $w(y-x)$.

In the case of convolutional codes, one needs an analogous notion for
polynomial vectors 
 $x(z)=(x_1(z),\dots,x_n(z))\in\F_q[z]^n\,.$
 If one writes $x(z)\in\F_q[z]^n$ as a polynomial with vector coefficients,
$$
x(z)=\sum_t x_tz^{t}\,,\text{ where
$x_t=(x_{t1},\dots,x_{tn})\in\F_q^n$\,,}
$$
then one can define a natural notion of \emph{weight in
convolutional coding theory} as folows:

\begin{defin} The \emph{weight} of $x(z)\in\F_q[z]^n$ is 
$$
w(x(z))=\sum_t w(x_t)\,.
$$
\end{defin}

\begin{defin} The free distance of an $(n,k)$ convolutional code $\Cc\subseteq \F_q(z)^n$ is
$$
\df=\min \{w(x(z))\ |\ x(z)\in\Cc \cap\F_q[z]^n\,,\ x(z)\neq 0\}\,.
$$
\end{defin}

In particular, if the degree of the code is zero, $\Cc$ is a
linear code and the (free) distance is the (minimum) distance as
linear code.

One says  that a linear code $\Cc(n,k)$ is MDS if its Hamming distance attains the Singleton bound $n-k+1$. Analogously one has:

\begin{defin} A convolutional code is MDS if its free distance $\df$ attains the generalized  \emph{Singleton bound} \cite{RoSm99}; that is,
$$\df=(n-k)(\lfloor\delta/k\rfloor+1)+\delta+1\,,$$
where $n$, $k$ and $\delta$ are respectively the length, dimension and degree of the  convolutional code, $C(n,k,\delta)$.
\end{defin}
 
\bigskip

In order to compute the free distance of a one-dimensional convolutional code $C(n,1,\delta)$ in terms of the polynomial decomposition of a canonical generator matrix
$$\mathcal G=G_0+G_1z+G_2z^2+\cdots+G_{\delta}z^{\delta}\,,$$
it  is usefull  to know  the Hamming weights of the linear codes $G_j$, and  of the  linear codes $\begin{pmatrix}G_j\\\vdots\\ G_0\end{pmatrix}$ and $\begin{pmatrix}G_{\delta}\\\vdots\\ G_{\delta-j}\end{pmatrix}$, for all $0\leq j\leq \delta$.


\begin{thm}\label{linearG}
If the  codes $G_j$ are MDS for all $0\leq j\leq\delta$ and the codes $\begin{pmatrix}G_j\\\vdots\\ G_0\end{pmatrix}$ and $\begin{pmatrix}G_{\delta}\\\vdots\\ G_{\delta-j}\end{pmatrix}$ are $(n,j+1)$ linear codes and MDS for all $0\leq j\leq \delta$, then $\delta< n$, $C(n,1,\delta)$ is a MDS convolutional code, and $\mathcal G=G_0+G_1z+G_2z^2+\cdots+G_{\delta}z^{\delta}$ is a generator matrix.

\end{thm}
\begin{proof}
Note first that since $\begin{pmatrix}G_\delta\\\vdots\\ G_0\end{pmatrix}$ has rank $\delta+1$ and $n$ columns, one has $\delta<n$.

The polynomial codewords of the code $C(n,1,\delta)$ have the form $p_j(z)G$, where $p_j(z)\in\F_q[z]$ is a polynomial of degree  $j$. Thus,  $p_j(z)=a_0+a_1z+a_2z^2+\cdots+a_jz^j$, with $a_j\neq0$. Moreover, since the  weight does not change as we multiply by  $z$, we can assume that $a_0\neq0$. 

The polynomial coefficients of the codeword $p_j(z)\mathcal G$ are codewords of the linear codes considered in the statement; thus, a lower bound $I_j$ for the weight $w(p _j(z)\mathcal G)$ is given by the sum of their minimal distances. One has:
$$
\begin{aligned}
I_0&=n(\delta+1)\\
I_j&=2n+2(n-1)+2(n-2)+\cdots+2(n-j+1)+(n-j)(\delta-j+1)\,,\\
&\text{ \hskip 12pt for $j\leq \delta$}\\
I_{\delta+i}&=I_\delta+i(n-\delta)\,,\text{ for $i\geq0\,,$}
\end{aligned}
$$
which leads to
$$
I_j=(j+1)n+(n-j)\delta\,,\text{ for every $j\geq0$}\,.
$$
Since $I_{j+1}-I_j=n-\delta>0$, for $j\geq0$, the free distance of the code is
$$
\df(C(n,1,\delta))=I_0=n(\delta+1)\,.
$$
Thus, $C(n,1,\delta)$ is a MDS convolutional code.
\end{proof}

Convolutional Goppa codes provide examples of this situation, as we shall see later in this paper.

\section{Summary of Convolutional Goppa Codes Theory}\label{sec:goppa}

Let $(X,\Oc_X)$ be a smooth
projective curve over $\F_q(z)$ of genus $g$, and let  $\Sigma_X$ be the field of rational functions of $X$; we also assume that $\F_q(z)$ is algebraically closed in $\Sigma_X$. 

Given a set $p_1,\dots, p_n$ of $n$ different $\F_q(z)$-rational
points of $X$,  if $\Oc_{p_i}$ denotes the local ring at the point
$p_i$, with maximal ideal $\m_{p_i}$, and $t_i$ is a local parameter
at $p_i$, one has the exact sequences
\begin{equation}\label{point}
\begin{aligned}
0\to \m_{p_i} \to \Oc_{p_i}&\to \Oc_{p_i}/\m_{p_i}\simeq
\F_q(z)\to 0\cr s(t_i)&\mapsto s(p_i)\,.
\end{aligned}
\end{equation}
Let us consider the divisor $D=p_1+\dots+p_n$, with its associated
invertible sheaf $\Oc_X(D)$. One has an exact sequence of
sheaves
\begin{equation}\label{divisor}
0\to \Oc_X(-D)\to \Oc_X\to Q\to 0\,,
\end{equation}
where the quotient $Q$ is a sheaf with support at the points
$p_i$.

Let $G$ be a divisor on $X$ of degree $r$, with support disjoint
from $D$. Tensoring the exact sequence (\ref{divisor}) by the
associated invertible sheaf $\Oc_X(G)$, one obtains
\begin{equation}\label{sheaf}
0\to \Oc_X(G-D)\to \Oc_X(G)\to Q\to 0\,.
\end{equation}

For each divisor $F$ over $X$, let us denote their
$\F_q(z)$-vector space of global sections by
$$L(F)\equiv
\Gamma(X,\Oc_X(F))=\{s\in \Sigma_X\ |\ (s)+F\geq 0\}\,,$$ where
$(s)$ is the divisor defined by $s\in \Sigma_X$. Taking global
sections in (\ref{sheaf}), one obtains
$$
\begin{aligned}
0\to L(G-D) \to L(G)& \overset\alpha\to
\F_q(z)\times\overset{\overset{n}\smile}\dots\times \F_q(z)\to \dots\cr s&\mapsto
(s(p_1),\dots, s(p_n))\,.
\end{aligned}
$$

\begin{defin} \emph{(\cite{DMPS06},\cite{DMPS04})}The convolutional Goppa code $\Cc(D,G)$ associated
with the pair $(D,G)$ is the image of the $\F_q(z)$-linear map

\centerline{$\alpha\colon L(G)\to \F_q(z)^n$}

Analogously, given a subspace $\Gamma\subseteq L(G)$, one defines
the convolutional Goppa code $\Cc(D,\Gamma)$ as the image of
${\alpha}_{\vert_{\Gamma}}$.
\end{defin}

By construction, $\Cc(D,G)$ is a convolutional code of length $n$
and dimension
$$k\equiv\dim L(G)-\dim L(G-D)\,.$$

Under the condition $2g-2<r<n$, the
evaluation map $\alpha\colon L(G)\hookrightarrow \F_q(z)^n$ is
injective, and the dimension of $\Cc(D,G)$ is
$$k=r+1-g\,.$$

\emph{The dual convolutional Goppa code} of the code $\Cc(D,G)$ is the
$\F_q(z)$-linear subspace $\Cc^\perp(D,G)$ of $\F_q(z)^n$ given as  in  \ref{def:dual}.

As we proved in \cite[\S3]{DMPS06},
the dual convolutional Goppa code $\Cc^\perp(D,G)$ associated
with the pair $(D,G)$ is the image of the $\F_q(z)$-linear map
$\beta\colon L(K+D-G)\to \F_q(z)^n$, given by
$$\beta(\eta)=(\Res_{p_1}(\eta),\dots, \Res_{p_n}(\eta))\,,$$
where $K$ is  the canonical divisor of rational
differential forms over $X$.

Since we are taking $2g-2<r<n$, the map $\beta$ is injective, and
$\Cc^\perp(D,G)$ is a convolutional code of length $n$ and
dimension
$$\dim L(K+D-G)=n-(1-g+r)\,.$$

\subsection{Convolutional Goppa Codes over the projective line}
\
\smallskip

Let $X=\Pp^1_{\F_q(z)}=\Proj \F_q(z)[x_0,x_1]$ be the projective
line over the field $\F_q(z)$, and let us denote by $t=x_1/x_0$
the affine coordinate, by $p_0=(1,0)$  the origin point, and by $p_\infty=(0,1)$ the point at
infinity. 

Let us take  $p_1,\dots,p_n$ different rational points  of
$\Pp^1$ and the divisors
$$\begin{aligned}
&D=p_1+\dots+p_n\\
&G=r p_\infty - s p_0 \text{ , with
$0\leq s \leq r<n$}
\end{aligned}
$$
Since $g=0$, the evaluation map $\alpha\colon L(G)\to \F_q(z)^n$
is injective, and $\Img \alpha$ defines a convolutional Goppa code
$\Cc(D,G)$ of length $n$ and dimension $k=r-s+1$.

If $\alpha_i\in\F_q(z)$, $1\leq i\leq n$, is the local coordinate of the rational point 
$p_i\in\Pp^1_{\F_q(z)}$, so that
$$
\alpha_i=a_i z+b_i\text{,  with $a_i\neq0\,,\,b_i\in\F_q$,}
$$
then, the matrix of the evaluation map
$\alpha$ with respect to the basis $\{t^s,t^{s+1},\dots,t^r\}$ of
$L(G)$ is the following generator matrix for the code $\Cc(D,G)$:
\begin{equation}\label{matrixg}
\mathcal G=\left( \begin{matrix} \alpha_1^s & \alpha_2^s & \dots &
\alpha_n^s\cr \alpha_1^{s+1} & \alpha_2^{s+1} &\dots &
\alpha_n^{s+1} \cr \vdots & \vdots &\ddots & \vdots\cr \alpha_1^r
& \alpha_2^r & \dots & \alpha_n^r\cr
\end{matrix}\right)\,.
\end{equation}

We proved  in \cite[\S5]{DMPS06}  that the parity-check matrix  $H$,  with respect to the basis 
 $$
 \left\langle {\frac{dt}{t^s\displaystyle\prod_{i=1}^n(t-\alpha_i) }},
{\frac{t\ dt}{t^s\displaystyle\prod_{i=1}^n(t-\alpha_i) }}, \dots,
{\frac{t^{n-r+s-2} dt}{t^s\displaystyle\prod_{i=1}^n(t-\alpha_i)
}}\right\rangle\text{ of $L(K+D-G)$, }
$$
is:
\begin{equation}\label{matrixc}
H=\left( \begin{matrix} h_1  & h_2 &\dots & h_n\cr h_1\alpha_1 &
h_2\alpha_2 &\dots & h_n\alpha_n\cr \vdots & \vdots &\ddots &
\vdots\cr h_1\alpha_1^{n-r+s-2} & h_2\alpha_2^{n-r+s-2} & \dots &
h_n\alpha_n^{n-r+s-2}\cr
\end{matrix}\right)\,,
\end{equation}

where $h_j=\dfrac{1}{\alpha_j^s\displaystyle\prod_{\overset{i=1}{i\neq
j}}^n(\alpha_j-\alpha_i) }$.

\section{One-dimensional Convolutional Goppa Codes over $\Pp^1_{\F_q(z)}$ \label{sec:convgoppa}}

We shall construct  two families of examples of  one dimensional convolutional codes over $\Pp^1_{\F_q(z)}$ with canonical generator matrices \cite{McE98}, whose free distance, $\df$, attains the generalized Singleton bound, i.e., they are MDS convolutional codes \cite{RoSm99} (Theorems \ref{MDS} and \ref{MDS1}).

Let $L(G)=\langle t^s,t^{s+1},\dots,t^r\rangle\xrightarrow{\alpha} \F_q(z)^n$ be the evaluation map associated with the divisors  $D=p_1+\dots+p_n$ and  $G=r p_\infty - s p_0$, $0\leq s \leq r<n$, as in Section \ref{sec:goppa}.

Let us denote by $CGC(n,1)$ the one-dimensional convolutional Goppa code defined by restriction of the evaluation map $\alpha$ to the subspace 
$$\Gamma=\langle t^s+t^{s+1}+\dots+t^r\rangle\subseteq L(G)\,.$$ 

If  $\,a_i z+b_i\text{,  with $a_i\neq0\,,\,b_i\in\F_q$}$, is the local coordinate of the rational point 
$p_i\in\Pp^1_{\F_q(z)}$, a generator matrix for the code $CGC(n,1)$ is
$$
\mathcal G=\begin{pmatrix}
\displaystyle\sum_{i=s}^r (a_1z+b_1)^i & \displaystyle\sum_{i=s}^r (a_2z+b_2)^i & \ldots & \displaystyle\sum_{i=s}^r (a_n z+b_n)^i
\end{pmatrix}
$$

We shall describe particular cases of these codes and we shall prove that some of these codes are MDS convolutional codes.

\begin{thm}\label{MDS}
 If $s=r$, a generator matrix  of the code  $CGC(n,1)$ is:
$$
\begin{pmatrix}(a_1z+b_1)^r&(a_2z+b_2)^r&\dots&(a_nz+b_n)^r\end{pmatrix}\,.
$$
If we choose $a_i, b_i$  for each  $1\leq i\leq n$,  such that $\dfrac{b_i}{a_i}=c^{i-1}$,  where $c$ is a primitive element of $\F_q$, a generator matrix  of the code  $CGC(n,1)$ is:
$$
\mathcal G=\begin{pmatrix}a_1^r(z+1)^r&a_2^r(z+c)^r&a_3^r(z+c^2)^r&\dots&a_n^r(z+c^{n-1})^r\end{pmatrix}\,.
$$
The matrix $\mathcal G$ is canonical and the code defined by $\mathcal G$ is MDS.
\end{thm}
\begin{proof}
$\mathcal G$ is clearly canonical. 

The memory and the degree $\delta$ of $CGC(n,1)$ are equal to $r$, so that the generalized Singleton bound is  $n(r+1)$.

The polynomial decomposition of $\mathcal G$ is
$$
 \begin{aligned}
 &\mathcal G=G_0+G_1z+G_2z^2+\cdots+G_rz^r\,, \text{ where}\\
 &G_j={r \choose j}\begin{pmatrix}a_1^r&a_2^rc^{r-j}&a_3^rc^{2(r-j)}&\cdots&a_n^rc^{(n-1)(r-j)}\end{pmatrix}\,,\text{  with $j=0,1,\dots,r$}\,,
 \end{aligned}
 $$
It is clear that for all $0\leq j\leq r$ the linear codes $G_j\,$, $\begin{pmatrix}G_j\\\vdots\\G_{0}\end{pmatrix}\,$, and $\begin{pmatrix}G_r\\\vdots\\G_{r-j}\end{pmatrix}$
are  MDS evaluation codes with  Hamming distances equal to $n$, $n-j$, and $n-j$, respectively. Then, as we proved in Theorem \ref{linearG}, the convolutional Goppa code $CGC(n,1)$ is MDS.
\end{proof}

\begin{thm} \label{MDS1} Let us consider the case when  $s=0$ (so that $\Gamma=\langle1+t+t^2+\cdots t^r\rangle$),  all  $b_i$'s are equal, $b_i=b$, and $a_i={a}^{i-1}$ where ${a}$ is  an  element of $\F_q$  with $\ord(a)\geq n$.
Then, the generator matrix of the code $CGC(n,1)$ 
$$
\mathcal G=\begin{pmatrix}
\displaystyle\sum_{i=0}^r (z+b)^i & \displaystyle\sum_{i=0}^r (az+b)^i & \ldots & \displaystyle\sum_{i=0}^r (a^{n-1} z+b)^i
\end{pmatrix}\,,
$$
is canonical and the code $CGC(n,1)$ is an MDS convolutional code of degree $r$ and free distance $n(r+1)$.
\end{thm}

\begin{proof}
If one denotes  $c_j= \sum_{m=j}^r {m \choose j} b^{m-j}$, $0\leq j\leq r$, the polynomial decomposition of $\mathcal G$ is
$$
 \begin{aligned}
 &\mathcal G=G_0+G_1z+G_2z^2+\cdots+G_rz^r\,, \text{ where}\\
 &G_j=c_j\begin{pmatrix}1&a^j&a^{2j}&\cdots&a^{(n-1)j}\end{pmatrix}\,,\text{  with $j=0,1,\dots,r$}\,.
 \end{aligned}
 $$
 
 The linear codes
$
L_{j,i}=\begin{pmatrix}G_i\\G_{i-1}\\\vdots\\G_{i-j}\end{pmatrix}, 0\leq j\leq i\leq r\,,
$
are also MDS evaluation codes with Hamming distance $d(L_{j,i})=n-j$, and hence $CGC(n,1)$ is a MDS convolutional Goppa code.
\end{proof}

\begin{remark}
In the case $b_i=0$, a generator matrix of the above MDS code is:
$$
\mathcal G=
\displaystyle\sum_{i=0}^r z^i \begin{pmatrix} 1& a^i & a^{2i}\ldots & a^{(n-1)i}
\end{pmatrix}
$$
and we obtain the class of one dimensional MDS convolutional codes of parameters $(n,1,r)$ constructed  by Gluesing and Langfeld \cite{HGBL06}.
\end{remark}

\medskip

\subsection{Classification of one-dimensional Convolutional Goppa Codes over the projective line}\label{sec:class}
\
\smallskip

Each one-dimensional convolutional Goppa code over $\Pp^1_{\F_q(z)}$ is defined by a subspace $\Gamma$ of dimension one of the vector space $L(G)$. We can therefore identify the set of one-dimensional convolutional Goppa codes with the projective space $\Pp(L(G))$ which is a variety of dimension $r-s$ over $\F_q$ since $L(G)=\langle t^s,\dots,t^r\rangle$.

Explicitly, each $CGC$ of dimension one is given by a generator matrix defined by:
$$
 \begin{aligned}
&\Gamma=\langle \lambda_st^s+\cdots+\lambda_r t^r\rangle\xrightarrow{\mathcal G} \F_q(z)^n\,,\,\lambda_i\in\F_q\\
&\mathcal G=\begin{pmatrix}
\displaystyle\sum_{i=s}^r \lambda_i(a_1z+b_1)^i & \ldots & \displaystyle\sum_{i=s}^r \lambda_i(a_n z+b_n)^i
\end{pmatrix}
 \end{aligned}
$$

Let us consider the case $s=0$. Here $L(G)=\langle 1,t,\dots, t^r\rangle$ and the set of $CGC$ of dimension one can be identified with the projective space $\Pp^r_{\F_q}$ of dimension $r$.

The condition for a code to be MDS is an open condition in $\Pp^r_{\F_q}$ (\cite{RoSm99} Lemma 4.1 and proof Theorem 2.10).
Since we have proved in Theorem \ref{MDS1} the existence of one dimensional convolutional Goppa codes of the MDS type, the set of MDS Goppa Codes of dimension one is a dense open subset of $\Pp^r_{\F_q}$ (considering $\Pp^r_{\F_q}$ as an algebraic variety). Essentially, this means that almost all one-dimensional $CGC$ are of MDS type.

In a forthcoming paper we shall give an explicit characterization of $CGC$ of dimension one that are of MDS type.



\begin{thebibliography}{1}
\providecommand{\url}[1]{#1}
\csname url@rmstyle\endcsname
\providecommand{\newblock}{\relax}
\providecommand{\bibinfo}[2]{#2}
\providecommand\BIBentrySTDinterwordspacing{\spaceskip=0pt\relax}
\providecommand\BIBentryALTinterwordstretchfactor{4}
\providecommand\BIBentryALTinterwordspacing{\spaceskip=\fontdimen2\font plus
\BIBentryALTinterwordstretchfactor\fontdimen3\font minus
  \fontdimen4\font\relax}
\providecommand\BIBforeignlanguage[2]{{%
\expandafter\ifx\csname l@#1\endcsname\relax
\typeout{** WARNING: IEEEtran.bst: No hyphenation pattern has been}%
\typeout{** loaded for the language `#1'. Using the pattern for}%
\typeout{** the default language instead.}%
\else
\language=\csname l@#1\endcsname
\fi
#2}}





\bibitem{DMPS06}
J.~M. Mu{\~n}oz~Porras, J.~A. Dom{\'{\i}}nguez~P{\'e}rez, J.~I. Iglesias~Curto and
  G.~Serrano~Sotelo, ``Convolutional codes of {G}oppa type,'' \emph{IEEE  Trans. Inform. Theory}, vol.~52, pp. 340--344, 2006.

\bibitem{DMPS04}
J.~A. Dom{\'{\i}}nguez~P{\'e}rez, J.~M. Mu{\~n}oz~Porras, and
  G.~Serrano~Sotelo, ``Convolutional codes of {G}oppa type,'' \emph{Appl.
  Algebra Engrg. Comm. Comput.}, vol.~15, no.~1, pp. 51--61, 2004.
  
\bibitem{DMS:08}
J.~A. {Dom\'{\i}nguez~P\'{e}rez}, J.~M. {Mu\~{n}oz~Porras}, and
G.~Serrano Sotelo, \emph{Algebraic geometry
  constructions of convolutional codes},
  in Advances in algebraic geometry codes,  pp.~365--391, World {S}cientific, May
  2008.
  
\bibitem{For70}
    G.D. Forney Jr, Convolutional codes I: Algebraic structure,
    \emph{IEEE Trans. Inform. Theory}, {\bf 16} (3), 720--738, (1970).

\bibitem{HGBL06}
H. Gluesing-Luerssen and  B. Langfeld, ``A class of one-dimensional MDS convolutional codes,'' \emph{Journal of Algebra and its Applications}, vol.~5, pp. 505--520, 2006.

\bibitem{McE98}
R.~J. McEliece, ``The algebraic theory of convolutional codes,'' in
  \emph{Handbook of coding theory, Vol. I}.\hskip 1em plus 0.5em minus
  0.4em\relax Amsterdam: North-Holland, 1998, pp. 1065--1138.
  
\bibitem{RoSm99}
J.~Rosenthal and R.~Smarandache, ``Maximum distance separable convolutional
  codes,'' \emph{Appl. Algebra Engrg. Comm. Comput.}, vol.~10, no.~1, pp.
  15--32, 1999.

\end{thebibliography}


%

%


\end{document}